\documentclass[letterpaper, 10 pt, conference]{ieeeconf}
\IEEEoverridecommandlockouts
\usepackage{cite}
\usepackage{amsmath,amssymb,amsfonts}
\usepackage{graphicx}
\usepackage{textcomp}
\usepackage{bbm}
\usepackage{mathrsfs}
\usepackage{algorithm}
\usepackage{algpseudocode}
\usepackage[dvipsnames]{xcolor}
\usepackage{tikz}
\usepackage{tikz-network}
\usepackage{hyperref}

\usepackage{etoolbox}

\newtheorem{lemma}{Lemma}

\newtheorem{definition}{Definition}
\newtheorem{problem}{Problem}
\newtheorem{proposition}{Proposition}

\newcommand{\Nc}{\mathcal{N}}

\newcommand{\Ac}{\mathcal{A}}
\newcommand{\Bc}{\mathcal{B}}
\newcommand{\Cc}{\mathcal{C}}
\newcommand{\Dc}{\mathcal{D}}
\newcommand{\Ec}{\mathcal{E}}

\newcommand{\Hc}{\mathcal{H}}

\newcommand{\Mc}{\mathcal{M}}

\newcommand{\Ic}{\mathcal{I}}

\newcommand{\Lc}{\mathcal{L}}
\newcommand{\Tc}{\mathcal{T}}

\newcommand{\Ns}{\mathscr{N}}

\newcommand{\Rb}{\mathbb{R}}
\newcommand{\Cb}{\mathbb{C}}

\newcommand{\Nb}{\mathbb{N}}
\newcommand{\Bf}{\mathfrak{B}}

\newcommand{\tr}{{\rm tr}}

\newcommand{\Span}{{\rm span}}

\newcommand{\vect}{{\rm vec}}
\newcommand{\ceil}[1]{\lceil #1 \rceil}

\begin{document}

\title{\LARGE \bf Reconstructing Quantum States from Local Observation:\\ A Dynamical Viewpoint}

\author{{Marco Peruzzo, Tommaso Grigoletto, Francesco Ticozzi
\thanks{The authors are with the Department of Information Engineering, University of Padova, Italy. Emails: \texttt{marco.peruzzo@dei.unipd.it}, \texttt{tommaso.grigoletto@unipd.it}, \texttt{ticozzi@dei.unipd.it}.}
}}

\maketitle
\thispagestyle{empty}
\pagestyle{empty}

\begin{abstract} 
We analyze the problem of reconstructing an unknown quantum state of a multipartite system from repeated measurements of local observables. In particular, via a system-theoretic observability analysis, we show that, even when the initial state is not uniquely determined for a static system, this can be reconstructed if we leverage the system's dynamics.
The choice of dynamical generators and the effect of finite samples is discussed, along with an illustrative example.
\end{abstract}

\section{Introduction}
Reliably estimating quantum states and dynamics from data is a key task in the development of quantum information technologies \cite{paris-estimationbook,zorzi2014}. The task becomes daunting if the system complexity increases, as the size grows exponentially with the number of components. It is then natural to look at estimation protocols that use a limited number of data. In general,  a quantum state on a multipartite system is not uniquely determined by its
marginals (i.e. reduced states), albeit this is indeed possible for specific classes of states \cite{reconstructinglocal}.
The problem of what is uniquely determined by local data has been widely studied from different perspectives \cite{linden_almost_2002,diosi_three-party_2004,eisert_gaussian_2008,yuCompleteHierarchyPure2021,xin_quantum_2017,weisQuantumMarginalsFaces2023,karuvade2019uniquely}, but so far the effect of the dynamics on these reconstruction problems has not been fully assessed, with some ideas presented in \cite{baioPossibleTimeDependent2018}.

In this work, we show how, by exploiting the system dynamics, one can retrieve the global information missing from local data, by effectively spreading the reach of local observables. This is done by resorting to the tools of observability analysis for linear systems specialized to the quantum domain 
\cite{domenicobook,grigoletto_2023,grigoletto_CDC}, and the construction of an appropriate output map: in fact, we establish equivalence between the ability to reconstruct every system's state and an observability condition on the model. Furthermore, we show that if we have a parametric family of dynamical generators, the observability property is generic as long as the dependency on the parameters is analytic. The methods are illustrated in a paradigmatic 4-qubit system undergoing different types of dynamics. 
In \cite{petersen2024,dynamicaltomography} similar tools have been adopted to perform quantum tomography and to establish an upper limit on the total number of measurements needed to reconstruct the state.

\section{Notation and problem
definition}\label{sec:preliminaries}

In this paper, we consider a multipartite quantum system composed of $N$ finite-dimensional subsystems. The associated Hilbert space has the following tensor product structure: 
\begin{equation*}
\Hc=\bigotimes_{q=1}^{N}\Hc_q, \ \dim(\Hc_q)=d_q, \ \dim(\Hc)=\prod_{q=1}^Nd_q=D\leq \infty.
\end{equation*}
The space of (bounded) linear operators acting on $\Hc$ will be denoted by $\Bc(\Hc)$ and the set of density operators by $\Dc(\Hc)=\{\rho | \rho=\rho^{\dag}\geq 0, \tr(\rho)=1\}$. We further denote by $\rho\in \Dc(\Hc)$ the state of the system. 

Multipartite systems such as the one we consider are often subject to locality constraints which restrict the structure of the possible evolution maps as well as the possible measurements that can be performed on the system at hand. The notion of locality can be formalized as follows \cite{karuvade2019uniquely,Johnson_2017,ticozzi2019whole}:
\begin{definition}[Neighborhood]
    A neighborhood $\Nc_k$ is a collection of sub-system indexes denoted as $\Nc_k\subseteq \{1\dots N\}$. A neighborhood structure $\Nc=\{\Nc_k\}_{k=1}^M$ is a finite collection of neighborhoods. We say that a neighborhood structure $\Nc$ is covering if every subsystem $j$ belongs to at least one neighborhood, i.e. $\cup_k \Nc_k = \{1,\dots, N\}$ and that it is non-trivial if $\{1\dots N\}\notin\Nc$.
\end{definition}
A neighborhood structure is \textit{connected} if for any two subsystems $i$ and $j$ there exists an ordered sequence of indexes (a path) that starts with $i$ and ends with $j$ with subsequent indices that belong to the same neighborhood. 
The \textit{complement} $\overline{\Nc}_k$ of a neighborhood $\Nc_k$ is a set of sub-system indexes such that $\Nc_k\cup \overline{\Nc}_k=\{1\dots N\}$ and $\Nc_k\cap \overline{\Nc}_k=\emptyset$. 
In the following, we restrict our attention only to covering, nontrivial and connected neighborhood structures. An example is depicted in Figure \ref{fig:neigh_structure}.
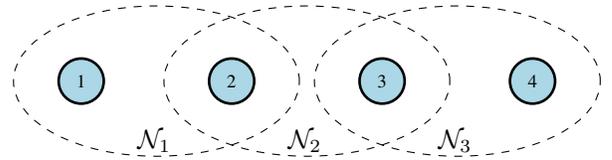
\begin{figure}[h!]
    \centering
    \begin{tikzpicture}
    \draw[dashed] (2,0) ellipse (1.9cm and 1cm);
    \draw[dashed] (4,0) ellipse (1.9cm and 1cm);
    \draw[dashed] (6,0) ellipse (1.9cm and 1cm);
    \node[text width=0.5cm] at (2,-0.8) {$\Nc_1$};
    \node[text width=0.5cm] at (4,-0.8) {$\Nc_2$};
    \node[text width=0.5cm] at (6,-0.8) {$\Nc_3$};
\Vertex[x=1,label=1]{1} \Vertex[x=3,label=2]{2}
\Vertex[x=5,label=3]{3}
\Vertex[x=7,label=4]{4}
\end{tikzpicture}
\caption{Example of a multipartite quantum system: a spin chain with 4 spins. 
The dashed lines highlight the considered neighborhood structure   $\Nc=\{ \Nc_1, \Nc_2, \Nc_3 \}$. 
}

\label{fig:neigh_structure}
\end{figure}

A neighborhood $\Nc_k$ defines a bipartition of the Hilbert space into 
\[\Hc = \Hc_{\Nc_k} \otimes \Hc_{\overline{\Nc_k}}\]
where $\Hc_{\Nc_k} = \bigotimes_{q\in\Nc_k} \Hc_q$ and $\Hc_{\overline{\Nc}_k} = \bigotimes_{q\in\overline{\Nc}_k} \Hc_q$. Operators that act trivially on the Hilbert space associated with the complement of a neighborhood $\Nc_k\in\Nc$ are said to be $\Nc$-\textit{local}, i.e. $X\otimes I_{\overline{\Nc_k}}$ with $X\in\Bc(\Hc_{\Nc_k})$ and $I_{\overline{\Nc_k}}\in\Bc(\Hc_{\overline{\Nc}_k})$. 

Given a quantum system and a neighborhood structure $\Nc$ we assume to be able to perform only measurements of $\Nc$-local observables on the system. These measurements allow us to compute the partial traces of the state $\rho\in\Dc(\Hc)$ on each neighborhood $\mathcal{N}_k\in \Nc$.
The set of all reduced density matrices associated with a neighborhood structure is defined as $\{\rho_{\mathcal{N}_k}=\tr_{\overline{\Nc}_k}(\rho)\}_{\Nc_k\in \Nc}$ where the super operators $\tr_{\overline{\Nc}_k}(\rho)$ indicate partial traces over $\Hc_{\overline{\Nc}_k}$.

The definitions of the marginal problem varies a little in the literature. We here follow the one given in \cite{baioPossibleTimeDependent2018}.
\begin{problem}[Quantum marginal problem (QMP)]\label{prb:QMP} 
Given a list of marginals $\{\rho_{\Nc_k}\}_{\Nc_k\in \Nc}$ characterize the set of states $\Mc_{\Nc}(\rho)\subset\Dc(\Hc)$ of the overall system that share this list of reduced density matrices, i.e. \[\rho_{\Nc_k} = \tr_{\overline{\Nc}_k}(\sigma), \quad \forall \sigma\in\Mc_\Nc(\rho). \]
\end{problem} 

While the reduced density matrices are uniquely determined given the state of the system, the converse is not true in general, hence the QMP for a certain list of marginals may admit multiple solutions. For this reason, the states that can be uniquely determined by their marginals are of particular interest \cite{karuvade2019uniquely}.
\begin{definition}[UDA States]\label{def:UDA}
A quantum state $\rho\in\Dc(\Hc)$ is uniquely determined among all states (UDA), with respect to a neighborhood structure $\Nc = \{\Nc_k\}_{k=1}^M$, if there exists no other state $\sigma \in \Dc(\Hc)$ with the same set of reduce density matrices, i.e. $\Mc_\Nc(\rho)\setminus\rho=\emptyset.$
\end{definition}
For simplicity, it is possible to represent the same information contained in the list of reduced density matrices $\{\rho_{\Nc_k}\}$ via a linear map $\Cc(\cdot)$ that projects the state $\rho$ on the subspace spanned by all the local operators associated to a neighborhood structure. In particular, consider an orthonormal Hermitian operator basis $\Bf_{\Nc_k} := \{E_{\Nc_k,j}\otimes I_{\overline{\Nc_k}}\}$ for the subspace $\Bc(\Hc_{\Nc_k})\otimes I_{\overline{\Nc_k}}$, i.e. the set of local operators associated with the neighborhood $\Nc_k$. Then we can define $\Cc[\rho(t)]=\sum_i C_i \tr(C_i \rho(t))$ where $\{C_i\}=\cup_k \Bf_{\Nc_k}$. Notice that this definition of the set $\{C_i\}$ implies that operators that act on the intersections between different neighborhoods are considered only once in the set, thus removing redundant information. 

Furthermore, if the neighborhood structure $\Nc$ is non-trivial then $\ker[\Cc] \neq \emptyset$: in fact, let $E_q$ be elements of $\Bc(\Hc_q)$ such that $\tr(E_q)=0$ 
and define $E:=E_1\otimes \dots \otimes E_N$ then $\Cc[E]=\sum_i C_i \tr(C_i E)=0$ since for any $C_i = X_{\Nc_k}\otimes I_{\overline{\Nc}_k}$ it holds $\tr(C_i E) = \tr(X_{\Nc_k}\bigotimes_{q\in\Nc_k}E_q)\prod_{q\in\overline{\Nc}_k}\tr(E_q) = 0$, i.e. $E$  is orthogonal to all the operators $C_i$. 

In practical applications, one is interested in uniquely reconstructing from locally constrained measurements the largest number of states of the system. This can be done for some classes of states (pure, maximally entangled) by properly choosing a covering neighborhood structure. However, any nontrivial neighborhood structure admits states that are not UDA.

\begin{lemma} \label{lem:non_triv_neigh}
    For any system with a non-trivial neighborhood structure, full-rank states are not UDA.
\end{lemma}
\begin{proof}
Since the neighborhood structure is nontrivial, 
$\ker[\Cc]\neq\emptyset.$ Given any $X\in\ker\Cc$ such that $\tr(X)=0$ and $\epsilon$ small enough so that, for any full-rank $\rho$, $\Bar{\rho}=\rho+\epsilon X$ is a state, i.e. $\Bar{\rho}\in \Dc(\Hc)$. This implies $\Cc[\Bar{\rho}]=\Cc[\rho]+\epsilon\Cc[ X]=\Cc[\rho]$ and therefore $\Bar{\rho}\in \Mc_\Nc(\rho)$. Then $\Mc_\Nc(\rho)\setminus \rho \neq \emptyset.$ 
\end{proof}

\section{The Role of the dynamics}
When the system of interest is in evolution the dynamics represents an additional degree of freedom that can be exploited in the solution of the QMP. As illustrated in the following section, adopting a dynamical viewpoint helps to significantly improve the set of states that can be uniquely determined given the marginals. In particular it becomes feasible to uniquely reconstruct all the states even for multipartite systems with non-trivial neighborhood structures.

In this work, we consider finite-dimensional Markovian, discrete-time quantum dynamics obtained by the discretization of continuous-time generators. Discrete-time quantum dynamics are described by a linear, completely positive (CP), and trace-preserving (TP) map $\Tc(\cdot)$ which guarantees that any element in $\Dc(\Hc)$ is mapped to another element in $\Dc(\Hc)$. Every CP map admits a representation in terms of an operator-sum (Kraus representation \cite{Nielsen2010}), i.e. $\Tc(\cdot)=\sum_k M_k \cdot M_k^\dag,$ where $ \{M_k\}\subset\Bc(\Hc)$, moreover the map is TP if and only if satisfies $\sum_k M_k^\dag M_k=I$.

A continuous semi-group of CPTP maps $\{\Tc_{t}\}_{t\geq 0}, \ \Tc_{0}=\Ic$ with the Markov composition property $\Tc_{t} \circ \Tc_{s}=\Tc_{t+s} \ \forall t,s\geq 0$ is called a Quantum Dynamical Semi-group (QDS) \cite{alicki2007semigroups}. Let $\Lc$ be the corresponding semi-group generator, i.e. $\Tc_{t}=e^{\Lc t}$. $\Lc$ can be expressed in Lindblad \cite{gks1979, lindblad1976} canonical form as
\begin{equation}\label{eq:limblad}
 \Lc(\cdot)=-i [H,\cdot] +\sum_k\left(L_{k} \cdot L_{k}^\dag-\frac{1}{2}\left\{L_{k}^\dag L_{k},\cdot\right\}\right),
\end{equation}
where $H=H^\dag$ is the time-invariant Hamiltonian of the system and the operators $L_{k}$, called noise operators, are the non-Hamiltonian components of the generators whose effect is to make the dynamics non-unitary and irreversible.

In the rest of the paper, we assume to have perfect knowledge of the dynamic, i.e. of the Hamiltonian $H$ and noise operators $L_k$.

As anticipated before, the generator $\Lc$ must account for the topological constraints imposed by the Neighborhood structure $\Nc$. To this aim, we introduce the definition of $\Nc$-local generator.
\begin{definition}\cite{johnson2016}
    A QDS generator $\Lc$ is $\Nc$-local (local with respect to a given neighborhood structure $\Nc$) if it may be expressed as:
    \begin{equation}
        \Lc=\sum_j \Lc_{j}, \ \  \Lc_{j}=\Lc_{\Nc_j} \otimes \Ic_{\overline{\Nc}_j}
    \end{equation}
    where $\Lc_{\Nc_j}$ is a Lindblad generator acting on $\Bc(\Hc_{\Nc_j})$ and $\Ic_{\overline{\Nc}_j}$ is the identity superoperator on $\Bc(\Hc_{\overline{\Nc_j}})$.
\end{definition}
The above definition is equivalent to requiring that the operators $\{L_{k}\}$ and $H$ of $\Lc$ can be written as
\begin{equation}
    L_{k}=L_{k\Nc_j}\otimes I_{\overline{\Nc_j}}, \ \ H=\sum_j H_{j}, \ \ H_{j}=H_{\Nc_j}\otimes I_{\overline{\Nc}_j},
\end{equation}
a Hamiltonian $H$ of this form is called a $\Nc$-Local Hamiltonian.
Throughout the rest of the paper, we will assume that the neighborhood structure for the generator $\Lc$ and the linear function $\Cc$ is the same. 

In addition to the free dynamics, we are interested in performing instantaneous measures on the system at different times $t\in\{k\Delta t\}, \ k\in\Nb_{\geq 0}$, where $\Delta t$ is the sampling time, to this aim it is convenient to consider a discretized version of the continuous time system defined above. 
The overall discrete-time dynamical systems reads:
\begin{equation}\label{eq:dis_sys}
    \Sigma:=\begin{cases}
     {\rho}[k+1]=\Ec ( {\rho}[k])\\
        \tau[k]=\Cc( {\rho}[k])
    \end{cases},
\end{equation}
where $\rho[k]:=\rho(k\Delta t)$, $\Cc$ is the output map which carries the information on the set of marginal of the state at time $t$, as defined before. $\Ec=\Tc_{\Delta t}=e^{\Lc \Delta t}$ is a CPTP map obtained by the discretization of the $\Nc$-local continuous dynamics. Such a model can be seen as an instance of a quantum hidden Markov model \cite{grigoletto_2023}.
An interesting question that arises when introducing dynamics into the quantum marginal problem is whether it is possible to reconstruct the initial state of the overall system, denoted as $ {\rho}(0)$.
 This leads to a dynamical version of Definition \ref{def:UDA}:
\begin{definition}
[UDDA states] \label{prb:DYNUD} A state $\rho\in \Dc(\Hc)$ is uniquely dynamically determined among all states (UDDA) if there does not exist any other state $\sigma\in \Dc(\Hc)$ such that $\Cc[\Ec ^k[ \rho]]=\Cc[\Ec ^k[ \sigma]] \ \forall k\in \Nb$.
\end{definition}
Our primary focus will be on establishing when the dynamics ensures that every state is UDDA from the marginals collected at multiple times.
To understand the solution to the aforementioned problem, we will leverage a well known tool, namely observability analysis.

By Definition \ref{prb:DYNUD} two initial states $ {\rho}_1,  {\rho}_2\in\Dc(\Hc)$ are dynamically indistinguishable (and thus non UDDA) if the corresponding outputs of the system are equal at all times. This is equivalent to the condition ${\rho}_1- {\rho}_2\in \ker{\Cc[{\Ec^k}]} \ \forall k\in \Nb$ by the linearity of the considered maps, which leads to the following definition.
\begin{definition}[Non-observable subspace]
The \text{non-observable} subspace for the system $\Sigma$ is the subspace
\end{definition}
\begin{equation}
\Ns :=\{ X \in \Bc(\Hc) \ | \ \Cc[\Ec ^k[ {X}]]=0, \forall k\in \Nb\}.
\end{equation}
By recalling that $\Cc[\rho]=\sum_i C_i\tr[C_i \rho]$, we can compute the orthogonal complement to $\Ns$ in $\Bc(\Hc)$, with respect to the standard Hilbert-Schmidt inner product:
\begin{equation}
\Ns ^\perp=span\{\Ec ^{k\dag}[C_i], \forall k\in\Nb, \  \forall i\}.
\end{equation}
We call $\Ns^\perp$ the \textit{observable subspace}.

Since the full operator space is of dimension $D^2$, there exists an integer value $k^*\leq D^2-1$ such that the observable subspace is completely characterized by considering the propagators only up to time $k^*$ \cite{kalman1969topics, wonham}, leading to
\begin{equation} \label{eq:obs_sub_char}
   \Ns ^\perp=span\{\Ec ^{k\dag}[C_i], \ k=0,1,\dots, k^*, \  \forall i\}.
\end{equation}
\begin{definition}
    The system is said to be observable if the  observable subspace dimension is $D^2$, i.e $\Ns^\perp=\Bc(\Hc)$.
\end{definition}

On the other hand, a lower-bound on the number of needed steps $k^*$ to check observability is also easily derived. If the system is observable and $m=|\{C_i\}|$, the minimum number of steps $k^*$ is such that $k^*\geq \ceil{D^2/m}$. In fact due to observability, through \eqref{eq:obs_sub_char} it must be possible to find $D^2$ independent generators of $\Ns$. However, if only $k\leq\ceil{D^2/m}-1$ steps are considered, the maximum number of linearly independent observables would be $m(\ceil{D^2/m}-1)<D^2.$ 

We next show that checking observability, which can be done systematically, is equivalent to check when every state is UDDA. 
\begin{proposition}
Every initial state $\rho(0)\in \Dc(\Hc)$ is UDDA if and only if the system \eqref{eq:dis_sys} is observable.
\end{proposition}
\begin{proof}
    We assume in this first part of the proof that the system is observable and we prove that this implies all the states are UDDA. To a $D\times D$ matrix $B$ it is possible to associate a $D^2$ dimensional vector $b={\rm vec}(B)$ by a linear transformation that stacks the columns of $B$ one below the other so that the $(i, j)$ entry of the matrix $B$ is the $(j-1)D + i$ entry of the vector $b$. $B$ can then be simply recovered from $b$ by properly arranging the entries of $b$ in a matrix. Some important properties of the aforementioned linear transformation are the following \cite{gilchrist2009vectorization}:
\begin{enumerate}
    \item Let $B,C$ be $D\times D$ dimensional matrices $\vect(ABC)=(C^T \otimes A) \vect(B),$ where $\otimes$ is the kronecker product;
    \item $\tr(A^\dag B)=\vect(A)^\dag \vect(B).$
\end{enumerate}
To the density operator $\rho[k]$ we thus associate a vector $r[k]=\vect(\rho[k])\in\Cb^{D^2}$ and to the output $\sigma[k]$ we associate the vector $y[k]=\vect(\tau[k])\in\Cb^{D^2}$.
By noting that 
\begin{subequations}
\begin{equation}\label{eq_vec_maps_a}
    \vect(\Ec[\rho[k]])=(\sum_k (M_k^\dag)^T \otimes M_k) r[k]=\hat{E} r[k],\\
\end{equation}
    \begin{equation}\label{eq_vec_maps_b}
        \vect(\Cc[\rho[k]])=\sum_i \vect(C_i)\vect(C_i)^\dag r[k]=\hat{C}r[k].
    \end{equation}
\end{subequations}
to the superoperator $\Ec$ we associate the matrix $\hat{E}$, to  $\Cc$ we associate $\hat{C}$. Therefore applying a superoperator to $\rho[k]$ is equivalent to pre-multiplying $r[k]$ by the associated matrix  \cite{amshallem2015approaches}.
Equation \eqref{eq:dis_sys} can be alternatively written as a linear system with state $r[k]$:
    \begin{equation}\label{eq:dis_sys_vect}
    \Sigma_v:=\begin{cases}
     r[k+1]=\hat{E}r[k]\\
        {y}[k]=\hat{C} r[k]
    \end{cases},
\end{equation}
where we set $\Delta t=1$. The outputs of the system at time instants $t=\{0,\dots,D^2-1\}$ are given by
\begin{equation}\label{eq:output}
        \begin{bmatrix}    
            {y}[0]  \dots {y}[D^2-1] 
        \end{bmatrix}^T=\hat{O}r[0],
\end{equation}
where \begin{equation}\label{eq:obs_matrix}
    \hat{O}=\begin{bmatrix}
        \hat{C} & \hat{C}\hat{E} & \dots &  \hat{C}\hat{E}^{D^2-1}
    \end{bmatrix}^T
\end{equation}
is called \textit{observability matrix}. The system $\Sigma_v$ is observable if and only $\Sigma$ is observable and it is a classical result that the system $\Sigma_v$ is observable if and only if $\hat{O}$ has rank $D^2$ (Kalman rank condition) \cite{kalman1969topics}.
If we pre-multiply both sides of equation \eqref{eq:output} by $\hat{O}^T$ the initial state of the system can be uniquely determined as 
\begin{equation}\label{eq:exact_reconstr}
    r[0]=[\hat{O}^T \hat{O}]^{-1} O^T\begin{bmatrix}
        {y}[0]  \dots {y}[D^2-1] 
    \end{bmatrix}^T,
\end{equation}
where $\hat{W}=[\hat{O}^T \hat{O}]$ is called observability gramian and it is a $D^2\times D^2$ matrix which is invertible since full rank as it has the same rank of $\hat{O}$. This proves the first implication.

Assume now the system is not observable. Similarly to the proof of Lemma \ref{lem:non_triv_neigh} given a full-rank state $\rho$ it is always possible to find $X\in\Ns$ such that $X\neq 0, \tr[X]=0$ and $\epsilon > 0$ so that $\Bar{\rho}={\rho+\epsilon X}$ is a state. This implies $\Cc(\Ec^k(\Bar{\rho}))=\Cc[\Ec^k(\rho)]+\epsilon \Cc[\Ec^k( X)]]=\Cc[\Ec^k(\rho)] \ \forall k$, hence $\rho$ and $\Bar{\rho}$ are dynamically indistinguishable and $\rho$ can not be UDDA.
\end{proof}
As a final remark, we report a more direct method to check observability, avoiding having to vectorize the system or obtain $\Ec$ from the Lindblad generator $\Lc$. Thanks to the following lemma the observability of the system can be checked directly using the Lindblad generator \cite{chen1995sampled}.
\begin{lemma}
    If every couple of distinct eigenvalues $\lambda_i, \lambda_j$ of $\Lc$ for which $\Re[\lambda_i]=\Re[\lambda_j]$ are such that $\Im[\lambda_i-\lambda_j]\neq 2 \pi s/\Delta t$ $\forall s \in \Nb$ then:
    \begin{equation*}
        \textrm{span}\{\mathcal{L}^{\dag j}{(C_i)}, \ \ \forall i, \forall j \in \{0,\dots,D^2-1\}\} = \Bc(\Hc)
    \end{equation*}
    implies $\Sigma$ observable.
\end{lemma}
 The lemma suggest also that it is better to avoid  certain values of the sampling time in the measurement procedure. For these values the system is not observable even if the underlying continuous time dynamics is such. 
 
\section{Parametric Dynamics}
In practical application, one is often able to engineer different dynamics for the system of interest by acting on a finite set of parameters of the QDS generator given in equation \eqref{eq:limblad}. In order to dynamically uniquely determine all the states we need to choose the free parameters so that the system is observable. Accordingly, this section aims at giving a method for the selection of the parameters. 
In this section, we consider the system $\Sigma$, whose generator $\Lc$ depends analytically on a finite number of parameters $\alpha \in \Rb^K$. We label such generator as $\Lc_\alpha$ and the Hamiltonian and set of noise operators for the generator as $H_\alpha$ and  $\{L_{k,\alpha}\}$. The system selected by $\alpha$ will be denoted as $\Sigma_\alpha$, $\Sigma_{v,\alpha}$ will be its vectorized version and the non observable subspace for the system will be labeled as $\Ns_{\alpha}$.
We now introduce a lemma which will be exploited to prove the main result of this section.

Consider an $m\times n$ matrix $A_\alpha = [f_{jk}(\alpha)]$, with $f_{jk} : R^K \mapsto \Cb$, such that its real and imaginary parts $\Re(f_{jk}),\Im(f_{jk})$ are (real)-analytic, and let $\mathsf{r} =    {\rm max}_{\alpha\in\Cb^K} {\rm rank}(A_\alpha)$.  
We have the following lemma \cite{ticozzi2013steadystate}:
\begin{lemma}\label{lem:matrix_gen}
    The set $\Ac=\{\alpha\in\Rb^K | \rm{rank}(A_\alpha)<\mathsf{r}\}$ is such that $\mu(\Ac)=0$, where $\mu$ is the Lebesgue measure in $\Rb^K$.
\end{lemma}
A selection criterion for the free parameters of the system is then given by the following proposition:
\begin{proposition} Let $H_\alpha$, $L_{k,\alpha} \ \forall k$  be matrices such that 
each of their entries has both real and imaginary parts which are analytic functions on the parameter $\alpha\in \Rb^K.$ If $\exists \ \hat{\alpha}$ such that the system $\Sigma_{\hat{\alpha}}$ is observable, then the set $\Ac=\{\alpha \in \Rb^K \ | \ \Ns_{\alpha}\neq\emptyset \}$ is such that $\mu(\Ac)=0$.
\end{proposition}
\begin{proof}
    As in the proof of the previous Proposition 1, we can always vectorize the system $\Sigma_{\hat{\alpha}}$ to obtain
    $\Sigma_{v,\hat\alpha}$. Since the system is observable, the observability
    matrix $\hat{O}_{\hat{\alpha}}$ has full rank which is equal to 
    $D^2$. The imaginary and real parts of $\hat{O}_{{\alpha}}$
    are analytic functions on the variable $\alpha$
    since they are obtained by the sum, multiplication, exponentiation of the
    entries of $H_\alpha$ and $L_{k,\alpha}$, these are all operations which
    preserve analyticity. The set $\{\alpha \in \Rb^K \ | \ {\rm rank}(\hat{O}_\alpha)<D^2\}$
    correspond exactly to the set $\Ac$ and the fact that $\mu(\Ac)=0$ follows
    from Lemma \ref{lem:matrix_gen}.
\end{proof}
This proposition suggest that we can arbitrarily set the parameter values. If the system is not observable and there exist a choice of parameters that makes it observable, by changing the values at random, we shall find a set of parameters which guarantees observability with probability one. 
Once we have found the parameters, equation \eqref{eq:exact_reconstr} allow us to reconstruct the 
initial state of the system provided we know exactly the outputs. As already mentioned $\rho[k]$ can then be simply recovered from $r[k]$ by properly arranging its entries in a matrix. This reconstruction method is however not reliable if we want to reconstruct the states form noisy data, collected from a real physical experiment.

\section{Reconstructing the state from measurements}

If we want to estimate  $\tau[k]=\sum_i C_i \tr(C_i \rho[k])$ which is the output of the map $\Cc$ at time $k\Delta t$ , we first need to estimate the quantities $\langle C_i[k] \rangle=\tr(C_i\rho[k]) \ \forall i$, i.e. the expected value of the observable $C_i$ with state $\rho[k]$. In order to compute the 
estimate $\hat{\tau}_i$ of $\langle C_i[k] \rangle$ we need to let the system evolve until time $t=k\Delta t$, perform the measurement of $C_i$, collect the outcome (one of the eigenvalues of $C_i$), and reset the experiment. This procedure needs to be repeated multiple times (say P times) to get sufficiently accurate estimates. Let $c_{i,j}[k]$ be the outcome of the experiment at time $k\Delta t$, then $\hat{\tau}_i[k]$ is taken as the empirical mean of the outcomes for the experiment that is  $\hat{\tau}_i[k]=(\frac{1}{P}\sum_{j=1}^P c_{i,j}[k])C_i$. The complete estimate of $\tau[k]$ will be then given by $\hat{\tau}[k]=\sum_{i}\hat{\tau}_i[k]$.
Clearly $\hat{\tau}_i[k]\rightarrow \langle C_i[k]\rangle C_i$ with probability 1 as $P\rightarrow\infty$.
The procedure to get the outputs from a physical system is summarized in Algorithm \ref{alg:time_trace_proc}. Assuming observability, there would be a unique state compatible with the outputs, if these were known exactly. If we take $y[k]=\vect(\hat{\tau}[k])$ the initial state would exactly be given by equation \eqref{eq:exact_reconstr}.   
However to get the initial state of the system from experimental data, the reconstruction method illustrated above is in general not always reliable. In fact, 1) accurate estimates $\hat{\tau}_i[k]$ of $\langle C_i[k]\rangle C_i$ are only given by averaging over a large quantity of trials, which is often unfeasible since the number of observable $C_i$ we need to measure grows rapidly with the number and the dimension of subsystems; 2) data can be subject to experimental noise, leading to significant errors; 
3) using experimental data directly in equation \eqref{eq:exact_reconstr}, there are no guarantees that the reconstructed state is a physical state i.e. $\hat{\rho}(0)\in \Dc(\Hc)$; 4) The precision of the solution $\hat{\rho}(0)$ is not only affected by noisy data but also by errors that can be generated when solving \eqref{eq:exact_reconstr}.

When we do not have exact knowledge of the outputs, alternative reconstruction methods are available. In particular it is possible to set up an optimization problem to search the initial state $\hat{\rho}(0)$ which satisfies all the constraints imposed by the set of physical states and whose outputs matches the set of estimated $\{\hat{\tau}(t)\}$. State reconstruction problems from noisy data are often formulated as \textit{maximum-entropy} problems or \textit{maximum-likelihood} problems \cite{parisML, zorzi2014}. In this work we consider the maximum-entropy approach (the other case can be treated in an analogous fashion), therefore we take as cost function in the optimization problem the entropy for the state $\rho\in\Dc(\Hc)$ defined as $S(\rho)=-\tr(\rho  \log(\rho)).$  By maximizing it  we obtain the most mixed state among the ones which satisfies the imposed constraints. The considered optimization problem reads: 
\begin{equation}
\begin{split}
    \hat{\rho}(0)&={\rm argmax} \ S(\rho)\\
    s.t. & \ 
    \rho \in \Dc(\Hc)\\
    & \ \Cc(\Ec^k(\rho))=\hat{\tau}[k] \ \forall k \in  \{0\dots k^*\}
\end{split}
\end{equation}
If a prior information $\sigma$ on the initial state is available, it is possible to include it in the optimization problem by considering as cost function the (Umegaki’s) quantum relative entropy $S(\rho || \sigma)$ between $\rho\in\Dc(\Hc)$ and $\sigma\in\Dc(\Hc)$, defined as $S(\rho || \sigma)=-\tr(\rho \log \rho - \sigma \log \sigma)$.

The optimization problem can be solved with the approach presented in \cite{zorzi2014}, where the feasibility of the problem is also addressed: if no solution can be found from the data, the imposed constraints are relaxed and the solution to the optimization problem with the new constraints is found. With this method the reconstructed initial state is ensured to be a physical state even in the presence of noisy data.

   \begin{algorithm}
    \caption{Acquisition of the outputs}\label{alg:time_trace_proc}
    \textbf{Require:} The physical system of interest\\
    \textbf{Output:} The list of estimated outputs $\{\hat{\tau}[k]\}$
    \begin{algorithmic}[1]
    
   \State Prepare the experiment
    \While{$k\leq k^*$}
        \While{$i\leq |\{C_i\}|$}
            \While{$j\leq P$}
                \State Reset the experiment
                \State Let the system evolve until time $t=k\Delta t$
                \State Measure $C_i$ and collect the outcome $c_{i,j}[k]$
            \EndWhile
            \State $\hat{\tau}_i(t)=(\frac{1}{P}\sum_{j=1}^P c_{i,j}[k])C_i$
        \EndWhile
        \State $\hat{\tau}[k]=\sum_{i}\hat{\tau}_i[k]$
    \EndWhile
    \end{algorithmic}
    \end{algorithm}

\section{Example: A 4 Qubit System}
In this section we highlight the effectiveness of the dynamical viewpoint for the solution of the marginal problem through some examples.
We consider a multipartite quantum system composed of 4 qubits disposed on a line with nearest-neighbor interactions, $\Hc=\bigotimes_{q=1}^4\Hc_q, \ \Hc_q=\Cb^2, \ \Hc \simeq \Cb^{16}$. The system is depicted in figure \ref{fig:neigh_structure}. The notation $\sigma_j^\alpha$  will  denote the Pauli operator $\sigma^\alpha, \alpha \in \{0,x,y,z\}$ acting on the j-th qubit, that is, $\sigma_j^\alpha \equiv I \otimes \cdots \sigma^\alpha \otimes \cdots I$ and similarly for ladder operators. Several evolutions for the system of interest have been considered,  and the observability analysis for each dynamics have been performed numerically. The software was used to find the matrices associated to the evolution and output maps according to equations \eqref{eq_vec_maps_a} and \eqref{eq_vec_maps_b}, then the observability matrix \eqref{eq:obs_matrix} was  computed and the Kalman rank condition was used to asses if the system is observable and hence if all the states are UDDA. 
First, we introduce a purely unitary evolution, where the Hamiltonian of the system describing interactions between adjacent spins is 
\begin{equation}
    H=\sum_{i=1}^4 \alpha_i\sigma_i^x + \beta_i \sigma_i^y+ \gamma_i \sigma_i^z + \sum_{i=1}^3 \delta_i \sigma_i^x\sigma_{i+1}^x+\epsilon_i \sigma_i^z\sigma_{i+1}^z
\end{equation}
and the coefficients $\alpha_i,\beta_i,\gamma_i, \delta_i, \epsilon_i \in \Rb$ can be chosen.

\noindent {\bf Case 1:} When all the coefficient are set to 1 the system is not observable, in particular the non observable subspace has dimension 5 and it is spanned by the generators:
\begin{equation}
\begin{split}
    \Ns = \Span\{I \otimes \sigma_x \otimes I \otimes \sigma_x, &  \ \ \ I \otimes\sigma_x \otimes I \otimes \sigma_y, \\ 
    I\otimes \sigma_x \otimes I \otimes \sigma_z,
 & \ \ \ I \otimes \sigma_x \otimes \sigma_x \otimes \sigma_x, \\  
    I \otimes \sigma_x \otimes \sigma_x \otimes \sigma_y. & \ \ \ \}
\end{split}   
\end{equation}
If we follow the results of proposition \ref{lem:matrix_gen}, by taking {\em all} the parameters at random we would be able to assess if this non-observability is generic or not. In fact, it is sufficient to take $\gamma_4$ as a sample drawn from a Gaussian distribution with mean 0 and variance 1, and we leave all the other parameters unchanged: for almost all the samples the system becomes observable. Let $\hat{O}^k=[\hat{C}\ \hat{C}\hat{E}\ \dots \ \hat{C}\hat{E}^{k} ]^T$, since $\textrm{rank}[O^{k^*} ]=\textrm{rank}[O^{k} ] \ \forall k\geq{k^*}$, $k^*$ is the the minimum $k$ such that $rank[O^k]=rank[O^{k+1}]$. 30 realizations have been taken and the mean number of steps $k^*$ necessary to have observability resulted to be 15. 

{\bf Case 2:} We consider now a second possible dynamics for the system which encompasses dissipative terms. We choose a local Lindblad generator with noise operators:
\begin{equation}
    L_i=\eta_i \sigma_i^+ \ \ i\in \{1\dots4\}.
\end{equation} 
The Hamiltonian part of the dynamics is the same as before with all the coefficients set to 1.
It is sufficient to consider $\eta_i=1 \ \forall i$ and the system is observable, the number of steps necessary to have observability is $k^*=14$.

In both these examples, $k^*$ is much lower than the maximum number of steps required by the observability analysis that is $D^2-1=255$. Through a preliminary observability analysis, it is hence possible to reduce the number of measurement necessary to collect data from a physical system and speed up the procedure outlined in algorithm  \ref{alg:time_trace_proc}.

{\bf Case 3:} As a last example we relax all the assumptions on having a covering neighborhood structure for the output map and we explore what happens if we perform observations only on a single neighborhood, for instance $\Nc_2$. We consider the same dynamics of the previous example where we draw independently all the coefficients of the model from a Gaussian distribution with mean zero and variance 1. The system is still observable, the number of steps necessary to have complete observability increases with respect to the previous examples ($k^*=29$) however it is still much lower with respect to the maximum number of steps. Moreover if we experimentally collect data from a single neighborhood the number of observable that needs to be measured at each time decreases. By optimizing the selection and the trade-off between number of observables and time required to collect data it is hence possible to significantly speed up the data collection procedure. This highlights the potential of the dynamical viewpoint to the quantum marginal problem with respect to the standard approach.

\section{Conclusions}
In this work we establish a connection between dynamical observability and the problem of reconstructing  the state of a multipartite system from its marginals, highlighting how exploiting the dynamics allows one to complement the missing information by effectively spreading the available measurements reach to the whole operator space. A formula to reconstruct the state from perfect data can be derived from standard system-theoretic considerations; in the case of imperfect data, due for example to finite sample size, solving the reconstruction problem is equivalent to solving a variational tomography problem. The genericity of the observability character within a class of parametric models is established, ensuring that randomized parameter choices allow one to assess observability with probability one.
The results are tested on a prototypical example, showing how one can trade the number of needed measurements with the number of samples needed for each trajectory. Further development of the study entails more general dynamics and measurement sets, imperfect knowledge of the dynamics, as well as methods to optimize the choice of parameters in a family of dynamical generators in order to obtain observability in the fastest possible way. 
\section*{Acknowledgment}
The authors acknowledge partial funding from the European Union - NextGenerationEU, within the National Center for HPC, Big Data and
Quantum Computing (Project No. CN00000013, CN 1, Spoke 10) and from the European Union’s Horizon Europe research and innovation programme under the project
“Quantum Secure Networks Partnership” (QSNP, grant agreement No 101114043).
\bibliographystyle{IEEEtran}
\bibliography{bibliography}

\begin{thebibliography}{10}
\providecommand{\url}[1]{#1}
\csname url@samestyle\endcsname
\providecommand{\newblock}{\relax}
\providecommand{\bibinfo}[2]{#2}
\providecommand{\BIBentrySTDinterwordspacing}{\spaceskip=0pt\relax}
\providecommand{\BIBentryALTinterwordstretchfactor}{4}
\providecommand{\BIBentryALTinterwordspacing}{\spaceskip=\fontdimen2\font plus
\BIBentryALTinterwordstretchfactor\fontdimen3\font minus \fontdimen4\font\relax}
\providecommand{\BIBforeignlanguage}[2]{{%
\expandafter\ifx\csname l@#1\endcsname\relax
\typeout{** WARNING: IEEEtran.bst: No hyphenation pattern has been}%
\typeout{** loaded for the language `#1'. Using the pattern for}%
\typeout{** the default language instead.}%
\else
\language=\csname l@#1\endcsname
\fi
#2}}
\providecommand{\BIBdecl}{\relax}
\BIBdecl

\bibitem{paris-estimationbook}
M.~Paris and J.~Rehacek, \emph{Quantum State Estimation}, ser. Lecture Notes in Physics.\hskip 1em plus 0.5em minus 0.4em\relax Springer Berlin Heidelberg, 2004.

\bibitem{zorzi2014}
M.~Zorzi, F.~Ticozzi, and A.~Ferrante, ``Minimum relative entropy for quantum estimation: Feasibility and general solution,'' \emph{IEEE Transactions on Information Theory}, vol.~60, no.~1, pp. 357--367, 2014.

\bibitem{reconstructinglocal}
B.~Swingle and I.~H. Kim, ``Reconstructing quantum states from local data,'' \emph{Phys. Rev. Lett.}, vol. 113, p. 260501, Dec 2014.

\bibitem{linden_almost_2002}
N.~Linden, S.~Popescu, and W.~K. Wootters, ``\BIBforeignlanguage{en}{Almost {Every} {Pure} {State} of {Three} {Qubits} {Is} {Completely} {Determined} by {Its} {Two}-{Particle} {Reduced} {Density} {Matrices}},'' \emph{\BIBforeignlanguage{en}{Physical Review Letters}}, vol.~89, no.~20, p. 207901, Oct. 2002.

\bibitem{diosi_three-party_2004}
L.~Diósi, ``Three-party pure quantum states are determined by two two-party reduced states,'' \emph{Physical Review A}, vol.~70, no.~1, p. 010302, Jul. 2004, publisher: American Physical Society.

\bibitem{eisert_gaussian_2008}
J.~Eisert, T.~Tyc, T.~Rudolph, and B.~C. Sanders, ``Gaussian quantum marginal problem,'' \emph{Communications in mathematical physics}, vol. 280, pp. 263--280, 2008.

\bibitem{yuCompleteHierarchyPure2021}
X.-D. Yu, T.~Simnacher, N.~Wyderka, H.~C. Nguyen, and O.~G{\"u}hne, ``A complete hierarchy for the pure state marginal problem in quantum mechanics,'' \emph{Nature Communications}, vol.~12, no.~1, p. 1012, 2021.

\bibitem{xin_quantum_2017}
T.~Xin, D.~Lu, J.~Klassen, N.~Yu, Z.~Ji, J.~Chen, X.~Ma, G.~Long, B.~Zeng, and R.~Laflamme, ``\BIBforeignlanguage{en}{Quantum {State} {Tomography} via {Reduced} {Density} {Matrices}},'' \emph{\BIBforeignlanguage{en}{Physical Review Letters}}, vol. 118, no.~2, p. 020401, Jan. 2017.

\bibitem{weisQuantumMarginalsFaces2023}
S.~Weis and J.~Gouveia, ``The face lattice of the set of reduced density matrices and its coatoms,'' \emph{Information Geometry}, pp. 1--34, 2023.

\bibitem{karuvade2019uniquely}
S.~Karuvade, P.~D. Johnson, F.~Ticozzi, and L.~Viola, ``Uniquely determined pure quantum states need not be unique ground states of quasi-local hamiltonians,'' \emph{Physical Review A}, vol.~99, no.~6, p. 062104, 2019.

\bibitem{baioPossibleTimeDependent2018}
G.~Baio, D.~Chru{\'s}ci{\'n}ski, and A.~Messina, ``A possible time-dependent generalization of the bipartite quantum marginal problem,'' \emph{Journal of Russian Laser Research}, vol.~39, pp. 422--437, 2018.

\bibitem{domenicobook}
D.~D’Alessandro, \emph{Introduction to {Quantum} {Control} and {Dynamics}}, 2nd~ed., ser. Advances in Applied Mathematics.\hskip 1em plus 0.5em minus 0.4em\relax Chapman and Hall/CRC, 2021.

\bibitem{grigoletto_2023}
T.~Grigoletto and F.~Ticozzi, ``Model reduction for quantum systems: Discrete-time quantum walks and open markov dynamics,'' \emph{arXiv preprint arXiv:2307.06319}, 2023.

\bibitem{grigoletto_CDC}
------, ``Minimal resources for exact simulation of quantum walks,'' in \emph{2022 IEEE 61st Conference on Decision and Control (CDC)}, 2022, pp. 5155--5160.

\bibitem{petersen2024}
S.~Xiao, Y.~Wang, Q.~Yu, J.~Zhang, D.~Dong, and I.~R. Petersen, ``Quantum state tomography from observable time traces in closed quantum systems,'' \emph{Control Theory and Technology}, vol.~22, no.~2, pp. 222--234, May 2024.

\bibitem{dynamicaltomography}
M.~Kech, ``Dynamical quantum tomography,'' \emph{Journal of Mathematical Physics}, vol.~57, no.~12, dec 2016.

\bibitem{Johnson_2017}
P.~D. Johnson, F.~Ticozzi, and L.~Viola, ``Exact stabilization of entangled states in finite time by dissipative quantum circuits,'' \emph{Physical Review A}, vol.~96, no.~1, Jul. 2017.

\bibitem{ticozzi2019whole}
F.~Ticozzi, S.~Karuvade, and L.~Viola, ``The whole from the parts: Markovian stabilizing dynamics and ground-state cooling under locality constraints,'' in \emph{2019 IEEE 58th Conference on Decision and Control (CDC)}.\hskip 1em plus 0.5em minus 0.4em\relax IEEE, 2019, pp. 2310--2315.

\bibitem{Nielsen2010}
M.~A. Nielsen and I.~L. Chuang, \emph{Quantum Computation and Quantum Information: 10th Anniversary Edition}.\hskip 1em plus 0.5em minus 0.4em\relax Cambridge University Press, 2010.

\bibitem{alicki2007semigroups}
R.~Alicki and K.~Lendi, \emph{Quantum Dynamical Semigroups and Applications}.\hskip 1em plus 0.5em minus 0.4em\relax Springer Berlin Heidelberg, 2007, vol. 717.

\bibitem{gks1979}
V.~Gorini, A.~Kossakowski, and E.~C.~G. Sudarshan, ``{Completely positive dynamical semigroups of N‐level systems},'' \emph{Journal of Mathematical Physics}, vol.~17, no.~5, pp. 821--825, 05 1976.

\bibitem{lindblad1976}
G.~Lindblad, ``On the generators of quantum dynamical semigroups,'' \emph{Communications in Mathematical Physics}, vol.~48, no.~2, pp. 119--130, Jun. 1976.

\bibitem{johnson2016}
P.~D. Johnson, F.~Ticozzi, and L.~Viola, ``General fixed points of quasi-local frustration-free quantum semigroups: from invariance to stabilization,'' \emph{Quantum Information \& Computation}, vol.~16, no. 7–8, p. 657–699, may 2016.

\bibitem{kalman1969topics}
R.~E. Kalman, P.~L. Falb, and M.~A. Arbib, \emph{Topics in Mathematical System Theory}.\hskip 1em plus 0.5em minus 0.4em\relax McGraw-Hill New York, 1969, vol.~1.

\bibitem{wonham}
W.~M. Wonham, \emph{Linear Multivariable Control: a Geometric Approach}.\hskip 1em plus 0.5em minus 0.4em\relax Springer US, 1979.

\bibitem{gilchrist2009vectorization}
A.~Gilchrist, D.~R. Terno, and C.~J. Wood, ``Vectorization of quantum operations and its use,'' \emph{arXiv preprint arXiv:0911.2539}, 2009.

\bibitem{amshallem2015approaches}
M.~Am-Shallem, A.~Levy, I.~Schaefer, and R.~Kosloff, ``Three approaches for representing {Lindblad} dynamics by a matrix-vector notation,'' \emph{arXiv preprint arXiv:1510.08634}, Dec. 2015.

\bibitem{chen1995sampled}
T.~Chen and B.~A. Francis, \emph{Optimal Sampled-Data Control Systems}.\hskip 1em plus 0.5em minus 0.4em\relax Springer London, 1995.

\bibitem{ticozzi2013steadystate}
F.~Ticozzi and L.~Viola, ``Steady-state entanglement by engineered quasi-local markovian dissipation: Hamiltonian-assisted and conditional stabilization,'' \emph{Quantum Information \& Computation}, vol.~14, no. 3-4, pp. 265--294, 2014.

\bibitem{parisML}
K.~Banaszek, G.~M. D'Ariano, M.~G.~A. Paris, and M.~F. Sacchi, ``Maximum-likelihood estimation of the density matrix,'' \emph{Phys. Rev. A}, vol.~61, p. 010304, Dec 1999.

\end{thebibliography}
\end{document}